\documentclass[pra,showpacs,preprintnumbers,amsmath,amssymb,superscriptaddress]{revtex4}

\usepackage{bm}
\usepackage{graphicx}
\usepackage{amsbsy}
\usepackage{amsmath}
\usepackage{amsfonts}
\usepackage{amsthm}
\usepackage{listings}
\usepackage{enumerate}

\begin{document}

\theoremstyle{plain}
\newtheorem{theorem}{Theorem}
\newtheorem{lemma}[theorem]{Lemma}
\newtheorem{corollary}[theorem]{Corollary}
\newtheorem{conjecture}[theorem]{Conjecture}
\newtheorem{proposition}[theorem]{Proposition}
\newtheorem{remark}[theorem]{Remark}
\newtheorem{definition}[theorem]{Definition}
\newtheorem{example}[theorem]{Example}

\def\be{\begin{equation}}
\def\ee{\end{equation}}
\def\ba{\begin{align}}
\def\ea{\end{align}}
\def\bmx{\begin{pmatrix}}
\def\emx{\end{pmatrix}}
\def\bra{\langle}
\def\ket{\rangle}

\newcommand{\mB}{\mathcal{B}}
\newcommand{\mF}{\mathcal{F}}
\newcommand{\mE}{\mathcal{E}}
\newcommand{\mL}{\mathcal{L}}
\newcommand{\mI}{\mathcal{I}}
\newcommand{\mU}{\mathcal{U}}
\newcommand{\mN}{\mathcal{N}}
\newcommand{\mM}{\mathcal{M}}
\newcommand{\mS}{\mathcal{S}}
\newcommand{\mC}{\mathcal{C}}

\newcommand{\tr}{{\rm Tr}}
\newcommand{\trace}{{\rm Tr}}
\newcommand{\linearspan}{{\rm Span}}

\newcommand{\mbR}{\mathbb{R}}
\newcommand{\mbC}{\mathbb{C}}
\newcommand{\R}{\mathbb{R}}
\newcommand{\C}{\mathbb{C}}

\newcommand\kb[1]{|#1\rangle\langle#1|}
\newcommand{\eq}[1]{\begin{align}#1\end{align}}

\title{Additive bounds of minimum output entropies for unital channels and \\ an exact qubit formula}
\author{Motohisa Fukuda}\email{m.fukuda@tum.de}
\affiliation{Zentrum Mathematik M5,
Technische Universit\"at M\"unchen,
Boltzmannstrasse 3,
85748 Garching, Germany
} 

\author{Gilad Gour}\email{gour@ucalgary.ca}
\affiliation{Institute for Quantum Science and Technology and 
Department of Mathematics and Statistics,
University of Calgary, 2500 University Drive NW,
Calgary, Alberta, Canada T2N 1N4}

\date{\today}

\begin{abstract}
We investigate minimum output (R\'enyi) entropy of qubit channels and unital quantum channels.
We obtain an exact formula for the minimum output entropy of qubit channels,
and bounds for unital quantum channels. 
Interestingly, our bounds depend only on the operator norm
of the matrix representation of the channels on the space of trace-less Hermitian operators. 
Moreover, since these bounds respect tensor products,
we get  bounds for the capacity of unital quantum channels, 
which is saturated by the Werner-Holevo channel. 
Furthermore, we construct an orthonormal basis, besides the Gell-Mann basis,
for the space of trace-less Hermitian operators by using discrete Weyl operators.
We apply our bounds to discrete Weyl covariant channels with this basis,
and find new examples in which the minimum output R\'enyi $2$-entropy is additive.   
\end{abstract}

\maketitle

\section{Introduction}

\subsection{Preliminary}

\subsubsection{Maps on Hermitian matrices, norms and entropies}
 
Consider the real vector space of $n \times n$ Hermitian matrices denoted by $H_n$.
Let $H_{n,+}$ be the positive cone and
$H_{n,t}$ the affine space of matrices with trace $t$ in $H_n$.
Denote also by $H_{n,+,1} = H_{n,+} \cap H_{n,1}$ the set of density matrices (quantum states). 
We will use sometimes the capitals $K$ and $L$ to represent other spaces for notational convenience.

For a linear map $\Phi: H_n \to H_k$, we define the \emph{$p\to q$ norm} by
\be\label{eq:pq norm}
\| \Phi \|_{p \to q} \equiv \max_{0\not=\rho \in H_n} \frac {\|\Phi(\rho)\|_q}{\|\rho\|_p}. 
\ee
Note that in \cite{AmosovHolevo2003} it is shown that for $p=1$
\be\label{eq:AH result}
\| \Phi \|_{1 \to q}  = \max_{\rho \in H_{n,+,1}} \|\Phi(\rho)\|_q.
\ee
(Another norm defined for maps between complex vector spaces of matrices is studied briefly in Section \ref{sec:pq norm}.)

A quantum channel is a completely positive and trace-preserving map. 
For a quantum channel $\Phi$, we define the \emph{minimum output R\'enyi $\alpha$-entropy} by
\be
S_{\min,\alpha} (\Phi)= \min_{\rho \in H_{n,+,1}} S_\alpha (\Phi(\rho)).
\ee 
Here, $S_\alpha(\cdot)$ is the R\'enyi $\alpha$-entropy:
\be
S_\alpha (\sigma) = \frac 1{1-\alpha} \log (\trace \sigma^\alpha )
\ee
which is defined for $\sigma \in H_{k,+,1}$ and $0\leq\alpha\leq\infty$;
this is well-defined for $\alpha =1, \infty$ by taking limit. 
For $\alpha= 1$, it is the von Neumann entropy:
\be
S(\sigma) =- \trace [\sigma \log \sigma].
\ee
From now on, we write $S(\cdot) = S_1(\cdot)$ and $S_{\min}(\cdot)=S_{\min,1}(\cdot)$.
Importantly, $S_\alpha(\cdot)$ is non-increasing in $\alpha$.

For a channel $\Phi$, we can see easily that
\be\label{eq:equiv}
S_{\min,\alpha} (\Phi) = \frac \alpha {1-\alpha} \log \left( \|\Phi\|_{1\to \alpha} \right).
\ee
In general, it is difficult to calculate $S_{\min,\alpha}(\cdot)$ or $ \|\Phi\|_{1\to \alpha} $, but
we get an exact formula for qubit channels in Section \ref{sec:qubit} and 
bounds for unital channels in Section \ref{sec:m-bound}.
Interestingly, these bounds are saturated for depolarizing channel and Werner-Holevo channel,
see Remark \ref{remark:gap}.


\subsubsection{Additivity and multiplicativity of channels}\label{sec:additivity}
Two channels $\Phi$ and $\Omega$ are said to be additive if
\be
S_{\min,\alpha} (\Phi \otimes \Omega) =S_{\min,\alpha} (\Phi) + S_{\min,\alpha} (\Omega).
\ee
This is equivalent to the multiplicativity:
\be
\|\Phi \otimes \Omega \|_{1\to \alpha} = \|\Phi\|_{1\to \alpha}\|\Omega \|_{1\to \alpha}.
\ee
The above equivalence can be seen from \eqref{eq:equiv}.
The additivity (or multiplicativity) of channels were conjectured for $p=1$ in \cite{KingRuskai2001} and 
for $1\leq p\leq \infty$ in \cite{AmosovHolevoWerner2000}. 
They were proved not to hold in general  for $1<p$ in \cite{HaydenWinter2008} and $p=1$ in \cite{Hastings2009}.  
See also \cite{FKM2010,BrandaoHorodecki2010,FK2010}.
For $p=1$ it was also shown to be locally additive in~\cite{GourFriedland2012}.
Proofs in terms of asymptotic geometric analysis are found in \cite{AubrunSzarekWerner2010,AubrunSzarekWerner2010,Fukuda2014}. 
Those  additivity/non additivity properties are important for communication theory, see \cite{Holevo2006}.

For $p=2$, no example is found for additivity violation. 
Not many additive example are found yet either. 
Entanglement breaking channels \cite{King2003e}, unital qubit channels \cite{King2002} and depolarizing channels \cite{King2003} are proved to be additive
as well as some examples in \cite{DFH2006}.
We add up other additive examples in Section \ref{sec:additive}. 

\subsubsection{Complementary channels}
This subsection contains facts we use in Section \ref{sec:DWCCC bound} and Section \ref{sec:additive}.
Complementary channels were investigated in relation to additivity questions in \cite{Holevo2005,KMNR2007}.
The idea of complementary channels is to 
swap the output and the environment spaces in the framework of Stinespring dilation theorem
to create another channel. 
For pure input states, a channel $\Phi$ and its complementary $\Phi^C$ share the same non-zero eigenvalues of output states.
Given a channel in the Kraus form:
$$\Phi(\rho) = \sum_i A_i \rho A_i^*,$$
we can define its complementary channel uniquely up to isomorphism:
$$\left( \Phi^C (\rho)\right)_{i,j} = \trace \left[ A_i \rho A_j^* \right].$$
Moreover, since $(\Phi \otimes \Omega)^C = \Phi^C \otimes \Omega^C$, we have for $1 \leq p \leq \infty$
\be
\| \Phi \otimes \Omega \|_{1\to p} = \left\| \Phi^C \otimes \Omega^C\right\|_{1\to p},
\ee
or equivalently,
\be
S_{\min,\alpha}( \Phi \otimes \Omega ) = S_{\min,\alpha}\left( \Phi^C \otimes \Omega^C \right).
\ee
This means that channels are additive if and only if so are their complementary channels:
\be\label{eq:complementN}
\| \Phi \otimes \Omega \|_{1\to p} =\| \Phi \|_{1\to p}\| \Omega \|_{1\to p} \Longleftrightarrow
\left\| \Phi^C \otimes \Omega^C\right\|_{1\to p} = \left\| \Phi^C\right\|_{1\to p}\left\|  \Omega^C\right\|_{1\to p},
\ee
or equivalently,
\be \label{eq:complementE}
S_{\min,\alpha} (\Phi \otimes \Omega) =S_{\min,\alpha} (\Phi) + S_{\min,\alpha} (\Omega)
\Longleftrightarrow
S_{\min,\alpha} (\Phi^C \otimes \Omega^C) =S_{\min,\alpha} (\Phi^C) + S_{\min,\alpha} (\Omega^C).
\ee
For more details, see \cite{Holevo2005,KMNR2007}. 
This concept of complementarity is applied in Section \ref{sec:DWCCC bound} and Section \ref{sec:additive}.

\subsection{Our results} 

\subsubsection{Qubit inputs}
We first consider a quantum channel  $\Phi: H_2 \to H_n$ with a qubit input. In this case, we found a closed formula in Theorem~\ref{ggg} for $\| \Phi \|_{1 \to 2}$, which also provides a closed formula for $S_{\min,2}(\Phi)=-\log \| \Phi \|_{1 \to 2}^{2}$. If in addition also the output space is 2-dimensional, then the formula for $\| \Phi \|_{1 \to 2}$
can be used to derive a closed formula for $S_{\min,\alpha}$ for any $0\leq \alpha\leq\infty$  (see corollary~\ref{gag}). We arrive at these formulas using the Bloch representation of a qubit. Since any pure qubit can be identified with a point on the 3-dimensional Bloch sphere, the optimization involved in the calculation of $\| \Phi \|_{1 \to 2}$ is relatively a simple one. However, if the input dimension is higher than $2$, the optimization is no longer over a three dimensional sphere, and therefore becomes more cumbersome. 

For a quantum channel $\Phi: H_2 \to H_n$, the closed formula for $\| \Phi \|_{1 \to 2}$, can also be used to derive an upper bound on the Holevo capacity of $\Phi$. The Holevo capacity is defined by
\be\label{eq:chi}
\chi(\Phi):=\max_{\{(\rho_i, p_i)\}_i} 
\left[ S\left(\sum_i p_i \Phi(\rho_i)\right) - \sum_i p_i S\left( \Phi(\rho_i)\right)\right]\;.
\ee
Therefore, for a quantum channel $\Phi$ with output dimension $n$
$$
\chi(\Phi)\leq\log(n)-S_{\min}(\Phi) \;,
$$
where it is known from the results in \cite{King2002} and \cite{Holevo2005a} that equality holds for unital qubit channels.
From theorem~\ref{lowerbound} we have $S_{\min}(\Phi)\geq g(\| \Phi \|_{1 \to 2}^{2})$ where the function $g$ is given in Eq.~\eqref{lb}. This lower bound is optimal in the sense that if another function $f$ satisfying $S_{\min}(\Phi)\geq f(\| \Phi \|_{1 \to 2}^{2})$, then $g(\| \Phi \|_{1 \to 2}^{2})\geq f(\| \Phi \|_{1 \to 2}^{2})$.
In particular, $g(\| \Phi \|_{1 \to 2}^{2})\geq -\log(\| \Phi \|_{1 \to 2}^{2})$; see Fig.~1.
Therefore, for a qubit-input channel $\Phi: H_2 \to H_n$ we obtain the following new upper bound for the Holevo
capacity:
\be\label{holcap}
\chi(\Phi)\leq\log(n)-g\left(\| \Phi \|_{1 \to 2}^{2}\right)
\ee
where the closed expression for $\| \Phi \|_{1 \to 2}$ is given in Theorem~\ref{ggg}, and the function $g$ is defined in Eq.~\eqref{lb}. This upper bounds becomes an equality for unital qubit channels. Eq.~(\ref{holcap}) holds for all channels $\Phi: H_m \to H_n$, but the closed expression for $\| \Phi \|_{1 \to 2}$ (see Theorem~\ref{ggg}) holds only for $m=2$.

\subsubsection{Multiplicative bounds and operational meanings}\label{sec:operational}

A trace-preserving linear map $\Phi:H_n \to H_k$ is called unital if 
\be
\Phi(I_n/n) = I_k/k.
\ee
In Section \ref{sec:general formula}, we derive lower bounds for the minimum output R\'enyi $2$-entropy for unital quantum channels; we also derive bounds for $p \to q$ norms for general unital trace-preserving linear maps and apply it to unital quantum channels (i.e. completely positive maps). 
Interestingly, these bounds respect tensor products, which can be used to bound output R\'enyi $2$-entropy of tensor products of 
many unital channels, which is stated in Theorem \ref{theorem:bound2}. 
Importantly, these bounds are calculated from a function $\gamma(\cdot)$, given in \eqref{eq:gamma}. 
As you can see in \eqref{eq:formula A}, $\gamma(\cdot)$ depends on the operator norm of matrix representation of maps. 
Let us make some historical notes. 
In \cite{montanaro2013} an idea of multiplicative bounds is used to bound minimum output R\'enyi $\infty$-entropy of tensor products of channels. Other multiplicative bounds, for example the operator norm of partially transposed Choi matrices, were found in  \cite{FukudaNechita2004arxiv} to bound minimum output R\'enyi $2$-entropy of tensor products of channels. 

For the rest of this section, we deduce some operational meanings of such multiplicative bounds,
which correspond to Theorem \ref{theorem:bound2} in our paper.
The capacity $C(\cdot)$ is the maximum ratio in bits per channel use where information can be sent reliably with arbitrary small probability of errors.
The formula for the capacity was given in \cite{Holevo1998,SchumacherWestmoreland1997}:
\be\label{eq:cap}
C(\Phi) = \lim_{N \to \infty} \frac 1N   \chi \left(\Phi^{\otimes N}\right),
\ee
where $\chi(\cdot)$ is defined in \eqref{eq:chi}.
This immediately gives the following bound
\be 
C(\Phi) \leq \log k - \lim_{N \to \infty} \frac 1 N S_{\min} \left(\Phi^{\otimes N}\right).
\ee
for a unital channel $\Phi:H_n \to H_k$.
On the other hand, by using the monotonicity of R\'enyi $\alpha$-entropy  and Theorem \ref{theorem:bound2} we have
\be
S_{\min} \left(\Phi^{\otimes N}\right) 
\geq S_{\min,2} \left(\Phi^{\otimes N}\right) =  - \log \left(\left\|\Phi^{\otimes N}\right\|_{1\to2}^2\right) \geq N \log(\gamma(\Phi)).
\ee
We have proved:
\begin{theorem}[A bound for capacity of unital channels]\label{theorem:capacity}
Take a unital quantum channel $\Phi$. Then,
\begin{enumerate}
\item The regularized minimum output entropy has the following bound:
\[
\lim_{N \to \infty} \frac 1 N S_{\min}\left(\Phi^{\otimes N}\right) \geq - \log \gamma (\Phi).
\]
\item The capacity has the following upper bound. 
\[
C(\Phi) \leq \log k  + \log \gamma (\Phi).
\]
where $\gamma(\cdot)$ is defined in \eqref{eq:gamma}. 
\end{enumerate}
\end{theorem} 
Interestingly, these bounds turn out to be saturated by the Werner-Holevo channel. 
See Corollary \ref{corollary:dep}.

\subsubsection{Examples and applications}
After obtaining our general theory in Section \ref{sec:general formula}, for the rest of Section \ref{sec:m-bound},
we work on examples to show how to use it.
In Section \ref{sec:basis}, we introduce some orthonormal basis in $H_{n,0}$ made of \emph{discrete Weyl operators}.
Although, the Gell-Mann basis is quite famous, our new basis gets along with
\emph{discrete Weyl covariant channels}, which are introduced in Section \ref{sec:diagonal}. 
Previously, discrete Weyl covariant channels are investigated in \cite{DFH2006} to give a bound 
which is equivalent to Theorem \ref{theorem:gamma DWCC}, and found additive examples for $S_{\min,2}(\cdot)$. 
In fact, our Theorem \ref{theorem:bound2} can be seen as generalization of Theorem 2 in \cite{DFH2006}. 
Also, the complementary channels of discrete Weyl covariant channels are studied in Section \ref{sec:DWCCC bound}
when they are also unital. 
Based on this study, we give new additive examples for $S_{\min,2}(\cdot)$ in Section \ref{sec:additive}. 

In Section \ref{sec:pq norm}, we extend Lemma \ref{lemma:bound1} of Section \ref{sec:general formula} to 
the $p\to q$ norms when maps are defined on $M_n (\mathbb C)$. 
This is the usual definition of $p\to q$ norms, although our definition in \eqref{eq:pq norm}
is compatible with the concept of channels. 

\section{Minimum output entropy formula with qubit inputs} \label{sec:qubit} 

Let $\Phi:H_{2,+,1}\to H_{n,+,1}$ be a quantum channel, where $H_{n,+,1}$ is the set of density matrices acting on $\mbC^n$.
The minimum entropy output of $\Phi$ is defined by 
$$
S_{\min}(\Phi)\equiv\min_{\rho\in H_{2,+,1}}S\left(\Phi(\rho)\right)
$$
where $S(\rho)=-\tr\left(\rho\log\rho\right)$. Our goal here is to find a closed formula for this quantity.

For this purpose, we will calculate first the square of the maximum 2-norm of $\Phi$:
$$
\left(\| \Phi \|_{1 \to 2}\right)^{2} = \max_{\rho\in H_{n,+,1}}\|\Phi(\rho)\|_{2}^{2}=\max_{\rho\in H_{n,+,1}}\tr\left(\Phi(\rho)^2\right)
$$
It is well know that the optimal $\rho$ for both of the equations above is a pure state. In the qubit case, any pure state has the Bloch representation
$$
\rho=\frac{1}{2}\left(I+\vec{r}\cdot\vec{\lambda}\right)
$$
where $\vec{r}$ is a unit vector in $\mbR^3$, and $\vec{\lambda}=(\lambda_1,\lambda_2,\lambda_3)$ is a vector consisting of the three $2\times 2$ traceless Hermitian matrices
forming a basis for the three dimensional real vector space of $2\times 2$ traceless Hermitian matrices.
With this representation of $\rho$, 
$$
\left(\| \Phi \|_{1 \to 2}\right)^{2}=\max_{\vec{r}\in\mbR^3\;,\;\|\vec{r}\|=1}\frac{1}{4}\tr\left(\Phi(I)^2\right)+\frac{1}{2}\sum_{j=1}^{3}r_j\tr\left(\Phi(I)\Phi\left(\lambda_j\right)\right)
+\frac{1}{4}\sum_{j=1}^{3}\sum_{k=1}^{3}r_jr_k\tr\left(\Phi(\lambda_j)\Phi(\lambda_k)\right).
$$
Note that the matrix $A_{jk}\equiv\tr\left(\Phi(\lambda_j)\Phi(\lambda_k)\right)$ is a real symmetric matrix.
Therefore, for the purpose of our calculation, we will choose the basis
$\lambda_1,\;\lambda_2,\;\lambda_3$ to be the basis that diagonalizes the symmetric matrix $A_{jk}$; that is, w.l.o.g we assume that
$\tr\left(\Phi(\lambda_j)\Phi(\lambda_k)\right)=a_j\delta_{jk}$. Further, we denote by $b_j\equiv\tr\left(\Phi(I)\Phi(\lambda_j)\right)$. With these notations
$\left(\| \Phi \|_{1 \to 2}\right)^{2}$ is given by
\be\label{eq:qubit cal}
\left(\| \Phi \|_{1 \to 2}\right)^{2}=\frac{1}{4}\tr\left(\Phi(I)^2\right)+\frac{1}{4}\max_{\vec{r}\in\mbR^3,\;\|\vec{r}\|=1}\sum_{j=1}^{3}\left(2b_jr_j
+a_jr_{j}^{2}\right).
\ee
In order to maximize the function $f(\vec{r})\equiv\sum_{j=1}^{3}\left(2b_jr_j+a_jr_{j}^{2}\right)$ on the unit sphere we define the constrain function $g(\vec{r})=\|\vec{r}\|^2-1$ and use the Lagrange multipliers technique. Denoting by $\alpha\in\mbR$ the Lagrange multiplier, the condition $\bigtriangledown f=\alpha\bigtriangledown g$ gives the relation $$r_j=\frac{b_j}{\alpha-a_j}\;,$$
where the coefficient $\alpha$ is determined from the constraint equation
\begin{equation}\label{1}
\sum_{j=1}^{3}\frac{b_{j}^{2}}{(\alpha-a_j)^2}=1\;.
\end{equation}
We therefore obtain a closed formula for $\| \Phi \|_{1 \to 2}$, which we summarize in the following theorem.
\begin{theorem}[An exact formula for qubit inputs] \label{ggg}
Let $\Phi:H_{2,+,1}\to H_{n,+,1}$ be a quantum channel, then
\be\label{2}
\| \Phi \|_{1 \to 2}=\frac{1}{2}\sqrt{\tr\left(\Phi(I)^2\right)+\sum_{j=1}^{3}\frac{2\alpha-a_j}{(\alpha-a_j)^2}
b_{j}^{2}}\;,
\ee
with $a_j$ the eigenvalues of $A_{jk}\equiv\tr\left(\Phi(\lambda_j)\Phi(\lambda_k)\right)$, $b_j\equiv\tr\left(\Phi(I)\Phi(\lambda_j)\right)$, and $\alpha$ determined by Eq.~\eqref{1}.
\end{theorem}

Next, we state basic facts about qubit states:
\begin{proposition}\label{prop:qubit}
For qubit states, we have the following properties. 
\begin{enumerate}
\item Take $\rho \in H_{2}$, then $\rho$ is a pure state if and only if $\trace \rho =1$ and $\trace \rho^2 =1$.
\item For $\rho \in H_{2,+,1}$, we have 
\[
S_\alpha (\rho) = h_{2,\alpha} \circ f (\rho).
\]
Here,  $h_{2,\alpha}(x)=\frac{1}{1-\alpha}\left(x^{\alpha}+(1-x)^{\alpha}\right)$ is the binary $\alpha$-R\'enyi entropy, and 
\be\label{eq:f}
f(x) \equiv \frac {1+\sqrt{2x^2 -1}}{2}. 
\ee
\end{enumerate} 
\end{proposition}
The first statement is equivalent to the idea of the Bloch sphere, and the second statement leads to:
\begin{corollary}[An exact formula for qubit channels]\label{gag}
For a qubit channel $\Phi: H_{2,+,1}\to H_{2,+,1}$ the minimum $\alpha$-R\'enyi entropy output is given by
\be
S_{\min,\alpha}(\Phi)=h_{2,\alpha} \circ f\left(\| \Phi \|_{1 \to 2} \right) 
\ee
where and $\| \Phi \|_{1 \to 2}$ is given by~\eqref{2} and $f(\cdot)$ is defined in \eqref{eq:f}. 
\end{corollary}

While the formula above holds only for qubit channels, we can still use theorem~\ref{ggg} to derive a tight lower bound for $S_{\min}(\Phi)$ with $\Phi: H_{2,+,1}\to H_{n,+,1}$ and $n>3$. The lower bound follows from the following lemma.
\begin{lemma}
Let $\{p_j\}_{j=1}^{n}$ be a probability distribution. For a given fixed value of the ``index of coincidence",
$\sum_{j=1}^{n}p_{j}^{2}\equiv c$, the lowest possible value of the R\'enyi entropy is achieved by 
 the  probability distribution:
\be
\Big( \underbrace{\frac {1+\Delta}{1+k},\ldots,\frac {1+\Delta}{1+k}}_{k},\frac{1-k\Delta}{1+k} \Big)
\ee
and the value is given by 
\be\label{lb}
g_\alpha (c) \equiv \begin{cases} 
 \displaystyle \frac{1}{1-\alpha} \log \left(   k \left( \frac {1+\Delta}{1+k}\right)^\alpha + \left( \frac{1-k\Delta}{1+k}\right)^\alpha \right)
& \text{if }1<\alpha <2 \vspace{0.2cm} \\
\displaystyle \log(k+1)-\frac{1}{1+k}\Big[k(1+\Delta)\log(1+\Delta)+(1-k\Delta)\log(1-k\Delta)\Big] &
\text{if } \alpha=1
\end{cases} 
\ee
where $k\equiv\left\lfloor\frac{1}{c}\right\rfloor$ and $\Delta\equiv\sqrt{c-(1-c)k^{-1}}$. 
We write $g(\cdot) = g_1(\cdot)$. 
Moreover, $g(c)$ is a continuous and monotonically non-increasing function of $c$; see Fig.~1.
\end{lemma} 
Note that if $c\geq 1/2$ then $g(c)=h_{2}\left(\frac{1+\sqrt{2c-1}}{2}\right)$, where $h_2$ is the binary Shannon entropy.
The proof of the lemma above follows directly from the results in~\cite{HarremoesTopsoe2001,BerrySanders2003}.

\begin{figure}
\includegraphics[scale=0.35]{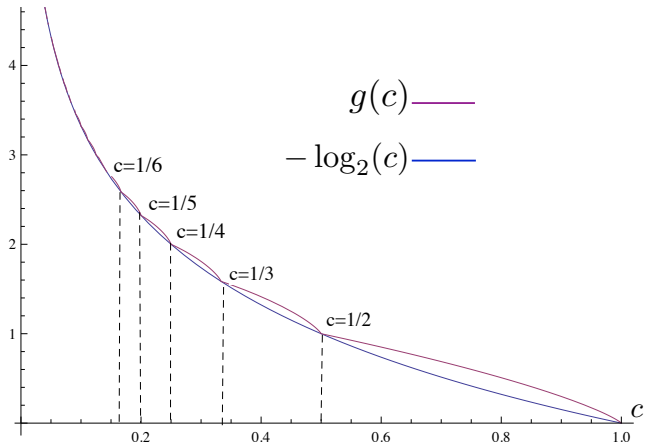}
\caption{The functions $g(c)$ and $-\log_2(c)$ are drawn in purple and blue, respectively. 
Note that $g(c)\geq -\log_{2}(c)$ with equality exactly at $c=1,1/2,1/3,1/4,...$. Furthermore, $g(c)+\log_2(c)\to 0$ in the limit $c\to 0$.}
\label{fig1}
\end{figure}

\begin{theorem} \label{lowerbound}
Let $\Phi:H_{m,+,1}\to H_{n,+,1}$ be a quantum channel, then
\be\label{lbone}
S_{\min,\alpha}(\Phi)\geq g_\alpha \left(\| \Phi \|_{1 \to 2}^{2}\right)
\ee
where the function $g_\alpha$ is defined in Eq.~\eqref{lb}, and equality holds for qubit output channels (i.e. if $n=2$).
\end{theorem}
\textbf{Remarks}:
\textbf{(1)} This lower bound is optimal in the sense that for any other function $f$ satisfying $S_{\min,\alpha}(\Phi)\geq f(\| \Phi \|_{1 \to 2}^{2})$, $g_\alpha(\| \Phi \|_{1 \to 2}^{2})\geq f(\| \Phi \|_{1 \to 2}^{2})$.
In particular, $g(\| \Phi \|_{1 \to 2}^{2})\geq -\log(\| \Phi \|_{1 \to 2}^{2})$ as shown in Fig.~1. \textbf{(2)} For $m=2$ and $n>2$, the closed expression given in Eq.~\eqref{2} can be used in Eq.~(\ref{lbone}) to obtain a tight lower bound on $S_{\min}(\Phi)$.

\section{Multiplicative bounds for unital channels}\label{sec:m-bound}
In this section, we leave qubit channels for higher dimensional cases. 
In Section \ref{sec:general formula},
we get bounds for the minimum output R\'enyi $\alpha$-entropy with $0 \leq \alpha \leq2$.
Considering the fact that getting exact values is difficult, it is interesting to get some bounds,
which are tight for some examples in Remark \ref{remark:gap}.
Also, our bonds are tight for some class of maps which include depolarizing channel and Werner-Holevo channel. 
Interestingly, Theorem \ref{theorem:bound2} our main theorem in this section
gives bounds for tensor products of unital channels. 
It's operational meanings are explained in Section \ref{sec:operational}.
Historically, a bound for the maximum $2$-norm of covariant channels are obtained in \cite{KMNR2007},
which was generalized in \cite{DFH2006} to tensor products of discrete Weyl covariant channels.
Below, we extend these ideas to unital channels. 
Towards the end of Section \ref{sec:general formula}, we define a sufficient condition for unital channels to have additive properties. In particular, with this observation, out bound for the capacity turns out to be tight for Werner-Holevo channel.

A new orthonormal basis made of discrete Weyl operators is given in Section \ref{sec:basis},
and applications of our bound to discrete Weyl covariant channels and their complementary channels are made in 
Section \ref{sec:DWCC bound} and Section \ref{sec:DWCCC bound}.
The definition of discrete Weyl covariant channels are given in Section \ref{sec:DWCC}.
Finally in Section \ref{sec:additive},  we give additive examples based on the additivity test developed in Section \ref{sec:general formula}.

\subsection{General formula}\label{sec:general formula}
Take an orthonormal basis in $H_{n,0}$: $\vec M = (M_1, \ldots, M_{n^2-1}) $
so that $(M_0=I_n/\sqrt n, M_1, \ldots, M_{n^2-1})$ forms an orthonormal basis of $H_n$
(an inner product on $H_n$ is defined by $\langle A, B\rangle=\tr(AB)$).
Take another space $K_{m}$ and then for any $\rho \in H_{n} \otimes K_{m}=L_{nm}$ we write
\be \label{eq:rho1}
\rho =   \sum_{i=0}^{n^2-1} M_i \otimes \rho_i = \frac {I_n } {\sqrt n} \otimes \rho_0  + \sum_{i=1}^{n^2-1} M_i \otimes \rho_i.
\ee
Here, 
$
\rho_i = \trace_H \left[ (M_i \otimes I_m) \rho \right] \in K_{m},
$
and in particular 
\be \label{eq:rho2}
\rho_0 = \trace_H \left[ \left(\frac{I_n}{\sqrt n} \otimes I_m \right)\rho  \right] = \frac{ \rho_K} {\sqrt n} 
\ee
with $\rho_K \equiv \trace_H [\rho] \in K_{m}$.
(Note that $\rho \in L_{nm,+,1}$ implies $\rho_K \in K_{m,+,1}$.)
Moreover, 
\be\label{eq:rho3}
\trace \left [ \rho^2\right] = \sum_{i=0}^{n^2-1} \trace \left[ \rho_i^2 \right].
\ee

Take a trace-preserving linear map $\Phi :H_n \to H_k$ and fix an orthonormal basis in $H_{k,0}$
to be $N_1, \ldots, N_{k^2-1}$. Since $\Phi$ is a linear map, 
it is written by a $(k^2) \times (n^2)$ real matrix denoted by $\tilde B_\Phi$.
In case $\Phi$ is unital, 
\be
\tilde B_\Phi = \begin{pmatrix} \sqrt{\frac nk} &0\\ 0&B_\Phi \end{pmatrix} 
\ee
where $B_\Phi: H_{n,0} \to H_{k,0}$.
To obtain the matrix $B_\Phi$, we calculate
\be
(B_\Phi)_{i,j} = \trace \left [ N_i \Phi(M_j) \right]  \in \mathbb R.
\ee
Then, we introduce a positive (symmetric) matrix: 
\be\label{eq:our matrix}
A_\Phi = B_\Phi^T B_\Phi : H_{n,0} \to H_{n,0}
\ee
which does not depend on choice of the above basis $N_1, \ldots, N_{k^2-1}$ in $H_{k,0}$. 
Indeed, we get the matrix $A_\Phi$ directly by
\be\label{eq:formula A}
(A_\Phi)_{i,j} = \trace \left[ \Phi(M_i) \Phi(M_j) \right].
\ee
To see this formula, 
\be
(A_\Phi)_{i,j} = \sum_{l=1}^{k^2-1} (B_\Phi)_{l,i}(B_\Phi)_{l,j} 
= \sum_{l=1}^{k^2-1} \trace \left[ N_l \Phi(M_i)  \right] \cdot \trace \left[ N_l \Phi(M_j)  \right] 
=\trace \left[ \Phi(M_i) \Phi(M_j)  \right].
\ee
Note that $N_1, \ldots, N_{k^2-1}$ are also orthonormal in the complex matrix space,
and the inner product is Euclidean.
In the analyses below, $\|A_\Phi\|_\infty$ plays a key role but this quantity 
does not depend on choice of basis $(M_0, \ldots, M_{n^2-1})$ in $H_{n,0}$ either. 

\begin{remark}
$A_\Phi$ can be diagonalized if we choose a proper basis in the domain of $\Phi$.
Such examples are studied in Section \ref{sec:DWCC bound} and Section \ref{sec:DWCCC bound}. 
\end{remark}

First, we have an important lemma as generalization of Theorem 1 in \cite{DFH2006}: 
\begin{lemma}\label{lemma:bound1}
For a trace-preserving unital linear map $\Phi: H_n \to H_k$ 
and for a Hermitian matrix $\rho \in H_n \otimes K_m$, 
\[
\trace \left[ \left(\Phi \otimes \mathrm{id} (\rho)\right)^2 \right] \leq 
\left( \frac 1 k - \frac{\|A_\Phi\|_\infty}{n}\right) \trace \left[\rho_K^2\right] + \|A_\Phi\|_\infty   \cdot  \trace \left[\rho^2\right],
\]
where $\mathrm{id}$ is the identity map on $K_m$. 
\end{lemma} 
\begin{proof}
By using the decomposition \eqref{eq:rho1}, we calculate
\eq{
\trace \left[\left( \Phi \otimes \mathrm{id} (\rho) \right) ^2 \right] 
&= \trace \left[ \left(  \Phi\left(\frac 1 n I_n\right) \otimes \rho_K  + \sum_{i=1}^{n^2-1} \Phi(M_i) \otimes \rho_i \right)^2\right] \\
&= \frac 1 k \trace \left[\rho_K^2\right] + \sum_{i,j=1}^{n^2-1} \trace \left[ \Phi(M_i) \Phi(M_j)\right] \cdot \trace \left[\rho_i \rho_j \right] .
\label{eq:tbound1}
}
Here, for the second inequality, we used the fact that $\Phi$ is unital and trace-preserving.
 
Since $A_\Phi \equiv B_\Phi^T B_\Phi$ is symmetric, by choosing $(M_i)_{i=1}^{n^2-1}$ properly, we can assume that 
$A_\Phi = \mathrm{diag} (\lambda_1, \ldots, \lambda_{n^2-1})$ 
with these entries non-increasing (remember \eqref{eq:formula A}).
Then, 
\eq{
\eqref{eq:tbound1}
&=  \frac 1 k \cdot \trace \left[\rho_K^2\right]   +  \sum_{i=1}^{n^2-1} \lambda_i \cdot \trace \left[\rho_i^2 \right]   
\leq \frac 1 k \trace \left[\rho_K^2\right]   +  \lambda_1 \sum_{i=1}^{n^2-1} \trace \left[\rho_i^2 \right] \label{eq:gap}\\
& = \left( \frac 1 k - \frac{\lambda_1}{n}\right) \trace \left[\rho_K^2\right] + \lambda_1   \sum_{i=0}^{n^2-1} \trace \left[\rho_i^2 \right]  
=\left( \frac 1 k - \frac{\lambda_1}{n}\right) \trace \left[\rho_K^2\right] + \lambda_1    \trace \left[\rho^2\right] . \label{eq:bound0}
}
Here we used \eqref{eq:rho2} and \eqref{eq:rho3}. 
To finish the proof, notice that
$\lambda_1 = \|A_\Phi\|_\infty$.
\end{proof} 

\begin{remark}\label{remark:our method}
Set $m=1$ in Lemma \ref{lemma:bound1}, then
the conditions $\trace \rho =1$ and $\trace \rho^2 = 1$ make 
 \eqref{eq:bound0} into
\be
\eqref{eq:bound0}=  \frac 1 k + \left(1-\frac 1 n \right) \| A_\Phi\|_\infty.
\ee
Here, again, $A_\Phi$ is given in \eqref{eq:formula A}.
Note that we did not use the positivity condition of $\rho$.
This fact shares spirits with Section \ref{sec:qubit}.
\end{remark}
Then, we define
\be\label{eq:gamma}
\gamma (\Phi) \equiv \begin{cases}
\displaystyle \frac 1k + (1-\frac 1n) \|A_\Phi \|_\infty & \text {if } n \geq k \cdot \|A_\Phi \|_\infty \vspace{0.2cm}\\ 
 \displaystyle\|A_\Phi \|_\infty & \text {if } n < k \cdot  \|A_\Phi \|_\infty .
\end{cases} 
\ee

With these notations we are ready to present the main theorem of this section: 
\begin{theorem}[Bounds for unital channels]\label{theorem:bound2}
We have the following bounds $0\leq\alpha\leq2$ and $N \in \mathbb N$. 
\begin{enumerate}
\item
Take two linear maps $\Phi$ and $\Omega$ defined between spaces of Hermitian matrices,
where $\Phi$ is unital and trace-preserving and $\Omega$ is completely positive.
Then, 
\[
\| \Phi \otimes \Omega \|_{1 \to 2} \leq \sqrt{\gamma (\Phi)} \cdot \|\Omega\|_{1\to 2}
\qquad\text{and}\qquad
S_{\min,\alpha} (\Phi \otimes \Omega) \geq -\log(\gamma(\Phi)) + S_{\min,2} (\Omega). 
\]
\item
For a sequence of unital channels $(\Phi_i)_{i=1}^N$, 
\[
\left\|  \bigotimes_{i=1}^N \Phi_i \right\|_{1\to2} \leq \quad \prod_{i=1}^N  \sqrt{\gamma(\Phi_i)}
\qquad\text{and}\qquad
S_{\min,\alpha}\left( \bigotimes_{i=1}^N \Phi_i \right)  \geq - \sum_{i=1}^N \log(\gamma(\Phi_i)).
\]
\end{enumerate}
\end{theorem} 
\begin{proof}
First, we prove the statements for the norms.
Suppose that $H_n$ and $K_m$ are domains of $\Phi$ and $\Omega$, respectively.
For $\rho \in H_n \otimes K_m = L_{nm}$  Lemma \ref{lemma:bound1} implies that 
\eq{\label{eq:bound2}
\trace \left[ \left(\Phi \otimes \Omega (\rho)\right)^2 \right] 
\leq \left( \frac 1 k - \frac{\|A_\Phi\|_\infty}{n}\right) \trace \left[\Omega(\rho_K)^2\right] 
+ \|A_\Phi\|_\infty   \cdot  \trace \left[\left(\mathrm{id} \otimes \Omega (\rho) \right)^2\right] .
}
On the other hand, for $\rho \in L_{nm,+,1}$
\[
\trace \left[\left(\mathrm{id} \otimes \Omega (\rho) \right)^2\right] \leq \| \Omega\|_{1 \to 2}^2,
\]
which  was proved in \cite{AmosovHolevoWerner2000}. 
This proves the first statement for because $\rho_K \in H_{n,+,1}$.
Note that when $k \cdot \|A_\Phi\|_\infty > n$ we ignore the first term in \eqref{eq:bound2} to get the bound.
The second statement is shown for by inductive applications of the first statement. 

Next, the above result can translate into the case of R\'enyi entropy with $\alpha=2$.
To complete the proof remember the monotonicity:
$
 S_{\alpha}(\cdot) \geq S_2(\cdot)
$
for $0\leq\alpha\leq2$.
\end{proof} 


\begin{remark} \label{remark:gap}
The above bounds in Lemma \ref{lemma:bound1} and Theorem \ref{theorem:bound2}
are saturated by the following class of linear maps on $H_n$:
\begin{equation}\label{eq:rescaling}
\Psi(\rho) = t\rho^\star + (1-t) \frac{\trace [\rho]} n I_n
\end{equation}
with $t \in \mathbb R$ such that $|t| \leq 1$, and $\star \in \{ 1, T\}$ where $T$ is transpose.
Special cases of these maps are called depolarizing channel or Werner-Holevo channels \cite{WernerHolevo2002}. 
Indeed, those bounds are saturated when the inequality in \eqref{eq:gap} has no gap,
 while the map $\Psi$ just rescales vectors by multiplying $t$ in $H_{n,0}$ up to rotations, so that $A_\Psi = |t|I$.
As a consequence, we have
\[
\| \Psi \otimes \Omega\|_{1\to2} =\| \Psi \|_{1\to2} \cdot \|  \Omega\|_{1\to2}.
\]   
for $\Psi$ in \eqref{eq:rescaling} and $\Omega$ a completely positive map.
See the proof of Theorem \ref{theorem:additivity}.
\end{remark}

Moreover, we can generalize the above additivity statement,
which is the generalization of Theorem 2 in \cite{DFH2006}.
To this end, we give the following definition. 
\begin{definition}\label{def:condition}
We define a condition called $\mathcal C_{add}$ in such a way that
a unital trace-preserving linear map $\Phi:H_n \to H_k$ satisfies $\mathcal C_{add}$ if 
 the following two equivalent conditions are satisfied. 
\begin{enumerate}
\item $\gamma(\Phi)=\|\Phi\|_{1\to2}^2$.
\item There exists a state which is supported within $H_{n,0}$ by the eigenspaces of the largest eigenvalue of $A_\Phi$ and 
$n \geq k \|A_\Phi \|_\infty$.
\end{enumerate}
The above equivalence is clear if we look into conditions when the inequality in \eqref{eq:gap} is saturated. 
\end{definition}

Then, as a corollary we have 
\begin{theorem}[Additivity test]\label{theorem:additivity}
Take two linear maps $\Phi$ and $\Omega$ defined between spaces of Hermitian matrices,
where $\Phi$ is unital, trace-preserving and satisfying $\mathcal C_{add}$ (Definition \ref{def:condition}), 
and $\Omega$ is completely positive.
Then, 
\[
\| \Phi \otimes \Omega \|_{1 \to 2}=\|\Phi\|_{1\to2} \cdot \|\Omega\|_{1\to 2}
\qquad\text{and}\qquad
S_{\min,2} (\Phi \otimes \Omega) = S_{\min,2} (\Phi) + S_{\min,2} (\Omega). 
\]
In particular, for a sequence of unital quantum channels $(\Phi_i)_{i=1}^N$ satisfying $\mathcal C_{add}$, 
\[
\left\|  \bigotimes_{i=1}^N \Phi_i \right\|_{1\to2} = \quad \prod_{i=1}^N \| \Phi_i\|_{1\to2}
\qquad\text{and}\qquad
S_{\min,\alpha}\left( \bigotimes_{i=1}^N \Phi_i \right)  = \sum_{i=1}^N S_{\min,2}( \Phi_i ).
\]
\end{theorem} 
\begin{proof}
We have 
\be
\| \Phi \|_{1\to2} \cdot \| \Omega \|_{1\to2} \leq \| \Phi \otimes \Omega\|_{1\to2} \leq \| \Phi \|_{1\to2} \cdot \| \Omega \|_{1\to2}. 
\ee
The first bound is trivial and the second comes from $\mathcal C_{add}$. 
\end{proof}

The following corollary justify our bound on the capacity in Theorem \ref{theorem:capacity}.
\begin{corollary}\label{corollary:dep}
Suppose $\Psi:H_n \to H_n$ is the Werner-Holevo channel:
\be
\Psi(\rho) = \frac{\trace[\rho]I_n-\rho^T}{n-1}
\ee
Then, the bound in Theorem \ref{theorem:capacity} is saturated:
\be
C(\Psi) = \log n + \log \gamma(\Psi).
\ee
\end{corollary}
\begin{proof}
By the contra-variant property of the channel we have \cite{Holevo2005a}
\be
C(\Psi) = \log n - \lim_{N\to\infty}\frac 1N  S_{min}\left(\Psi^{\otimes N}\right).
\ee
On the other hand, 
\be
S_{\min}\left(\Psi^{\otimes N}\right) \leq N S_{\min}\left(\Psi \right) 
= N S_{\min,2}\left(\Psi \right) = - N \log ( \gamma(\Psi)).
\ee
Here, the first equality comes from \cite{WolfEisert2005}, and 
the second holds because $\Psi$ satisfies $\mathcal C_{add}$.
Hence, 
\be
\log n + \log (\gamma(\Psi)) \leq C(\Psi). 
\ee
Then, Theorem \ref{theorem:capacity} completes the proof by sandwich theorem. 
\end{proof}

\subsection{A basis made of discrete Weyl operators in the real space of Hermitian matrices} \label{sec:basis}
To obtain the matrix $A_\Phi$ for a unital channel $\Phi$, 
we need to fix a basis in $H_{n,0}$, where $H_n$ is the domain of $\Phi$.  
For this purpose, Gell-Mann basis could be the first candidate. 
However, in this section we develop another orthonormal basis by using the discrete Weyl operators. 

Define the discrete Weyl operators by
\be\label{eq:DWO}
W_{x,y} = U^{x}V^{y} 
\ee
where $U,V \in \mathcal U(n)$ are unitary matrices as follows. 
\be\label{eq:UV}
U | m \ket =  |m+1\ket \qquad \text{and}\qquad V|m\ket = \exp \left( \frac {2 \pi \mathrm{i}} n \, m \right) | m \ket
\ee
for $m,x, y \in \{0,1, \ldots, n-1\} = \mathbb Z_n$, which is the group of integers modulo $n$. 
Note that for $(x,y)\not = (x^\prime,y^\prime)$ we have
\be
\trace \left[W_{x,y}^*W_{x^\prime,y^\prime}\right] =0.
\ee
In particular, $\trace W_{x,y}=0$ for  $(x,y)\not=(0,0)$.

Before introducing our basis, we set up four subsets of $\mathbb Z_n \times \mathbb Z_n$
by the following conditions. 
\begin{enumerate} [(i)]
\item $1 \leq x<y \leq n-1 $
\item $1\leq x=y\leq \left\lfloor \frac {n} 2 \right\rfloor$
\item $x=0, \quad 1 \leq y \leq \left\lfloor \frac {n} 2 \right\rfloor $
\item $y=0, \quad 1 \leq x \leq \left\lfloor \frac {n} 2 \right\rfloor $ 
\end{enumerate}
where $\lfloor r \rfloor$ is the largest integer such that $r \geq \lfloor r \rfloor$.
We define a set $\mathcal S_n=\{(x,y) \in \mathbb Z_n \times \mathbb Z_n: \text{$(x,y)$ satisfies (i), (ii), (iii) or (iv)}\}$. 
Now, set $\mathcal S_n^{F}= \{ \left(\frac n2,\frac n2\right),\left(0,\frac n2\right),\left(\frac n2,0\right)\} \subset \mathcal S_n$ and
 $\mathcal S_n^{GH} = S_n \setminus \mathcal S_n^F$.
Then, we define the following Hermitian matrices to construct an orthonormal basis in $H_{n,0}$: 
\begin{align} \label{eq:basis} 
\begin{aligned}
&\; \;F_{x,y} = \frac 1{\sqrt n} W_{x,y} 
&& \text{for $(x,y) \in \mathcal S_n^F $} \\
&\left. \begin{aligned}
&G_{x,y} = \frac 1 {\sqrt{2n}} \left(W_{x,y} + W_{x,y}^*\right) \\
&H_{x,y} = \frac 1 {\sqrt{2n} i} \left(W_{x,y} - W_{x,y}^*\right)  
\end{aligned}\right\} 
&&  \text{for $(x,y) \in \mathcal S_n^{GH} $.}
\end{aligned}
\end{align}
Note that $S_n^F = \emptyset$  when $n$ is odd. 

\begin{proposition}\label{proposition:orthonormal}
The matrices defined in \eqref{eq:basis} form an orthonormal basis in $H_{n,0}$.
\end{proposition} 
\begin{proof}
We start with counting the number of matrices defined in \eqref{eq:basis}. 
When $n$ is odd, 
\be
\binom{n-1}{2} \cdot 2 + \frac {n-1} 2 \cdot 2 + \frac{n-1} 2 \cdot 2 + \frac{n-1} 2 \cdot 2 = n^2-1.
\ee
When $n$ is even,
\be
\binom{n-1}{2} \cdot 2 
+ \left( \frac {n} 2 -1 \right)\cdot 2  + \left( \frac {n} 2 -1 \right)\cdot 2+ \left( \frac {n} 2 -1 \right)\cdot 2  + 3 = n^2-1.
\ee
Hence, we show that they are orthogonal. 
To this end, we note that 
\be
\trace \left[ W_{x,y} W_{x^\prime,y^\prime} \right] =c \trace \left[ W_{x+x^\prime,y+y^\prime}  \right] 
\ee
for some complex number $c$ with modulo $1$. 

First, we claim that $G_{x,y}$ and $H_{x,y}$ are orthogonal. 
Indeed, 
\be
\trace \left[( W_{x,y} + W_{x,y}^* )( W_{x,y} -  W_{x,y}^*)  \right] 
= \trace \left[ W_{x,y}^2  + (W_{x,y}^*)^2  \right] = \trace  \left[ cW_{2x,2y}  + \bar c W_{-2x,-2y}  \right] =0
\ee
where $c$ is some complex number with modulo $1$. 
The last equality holds for (i), (iii) and (iv) because
 $x\not=y$ implies that $2x$ and $2y$ can not be $0$ at the same time unless $(x,y) \in \mathcal S_n^F $.
For (ii), $2x=2y\not=0$ unless $(x,y) \in \mathcal S_n^F$, again. 

Next, we prove orthogonality when $(x,y)\not=(x^\prime,y^\prime)$. 
Orthogonality within (iii) and (iv) is clear because we have  $x+x^\prime\not=0$ or $y+y^\prime \not=0$
for $(x,y)\not=(x^\prime,y^\prime)$. 
Also, we know orthogonality between the first two cases ((i) and (ii)) and the last two cases ((iii) and (iv)).
Indeed, for $(x,y)$ from the first group and $(x^\prime, y^\prime)$ the second,
$x^\prime =0$ implies $x+x^\prime\not=0$ and $y^\prime=0$ $y+y^\prime \not=0$. 

Hence, we show orthogonality within (i) and (ii) as a whole to finish the proof. 
Take $(x,y),(x^\prime,y^\prime)$ from (i) and (ii).
Since $(x,y) \not =(x^\prime,y^\prime)$,
firstly $x+x^\prime =0$  implies $y+y^\prime \not=0$, and secondly $y+y^\prime =0$ implies $x+x^\prime \not=0$. 
\end{proof} 

%
%
%

\subsection{Weyl covariant channels as examples}\label{sec:diagonal}
%

\subsubsection{Discrete Weyl covariant channels}\label{sec:DWCC}
The discrete Weyl covariant channels are defined by
\be\label{eq:DWCC}
\Psi(\rho) =  \sum_{(x,y) \in \mathbb Z_n \times \mathbb Z_n} p_{x,y }W_{x,y} \rho W_{x,y}^* 
\ee
where $(p_{x,y})_{(x,y)}$ is a probability distribution. 
Remember $W_{x,y}$ are defined in \eqref{eq:DWO}.
The name comes from the property that 
\be
\Psi(W_{a,b} \, \rho \,W_{a,b}^*) = W_{a,b}\Psi( \rho ) W_{a,b}^* 
\ee
for all $(a,b) \in \mathbb Z_n \times \mathbb Z_n$. 
This is true because the discrete Weyl operators are commuting up to constants with modulo $1$. 
I.e., 
\be\label{eq:almost commuting}
W_{x,y} W_{a,b} 
=    c_{x,y,a,b} W_{a,b} W_{x,y}
\ee
where
\be\label{eq:phase}
c_{x,y,a,b}= \exp\left(\frac {2\pi i} n \left(a y - xb\right) \right).
\ee
See \cite{DFH2006} for more details about the discrete Weyl covariant channels. 
Not surprisingly, our matrix $A_\Psi$ in \eqref{eq:our matrix} is diagonal in our basis defined in Section \ref{sec:basis},
which you can see below. 

\subsubsection{How our formulas work with discrete Weyl covariant channels}\label{sec:DWCC bound}
\begin{lemma}\label{lemma:A diagonal}
For the channel $\Psi$ in \eqref{eq:DWCC}, the matrix $A_\Psi$ in \eqref{eq:formula A} is diagonal with respect to the basis in \eqref{eq:basis},
and the diagonal entries are given by $|c_{a,b}|^2$
where
\be\label{eq:phase2}
c_{a,b} =  \sum_{x,y}  p_{x,y}  c_{x,y,a,b}.
\ee
\end{lemma}
\begin{proof} 
\be 
\Psi (W_{a,b}^*) = \left(  \sum_{x,y} p_{x,y} W_{x,y} W_{a,b} W_{x,y}^*\right)^* 
= \sum_{x,y}  p_{x,y} \bar c_{x,y,a,b} W_{a,b}^*
\ee
where $ c_{x,y,a,b}$ is defined in \eqref{eq:phase}.
Hence, 
\be \label{eq:mapped basis}
\Psi(F_{a,b}) = \frac 1 {\sqrt n} c_{a,b} W_{a,b}, 
\quad
\Psi(G_{a,b}) = \frac 1 {\sqrt{2n}} \left(c_{a,b}W_{a,b} + \bar c_{a,b} W_{a,b}^*\right),
\quad 
\Psi(H_{a,b}) = \frac 1 {\sqrt{2n} i} \left(c_{a,b}W_{a,b} - \bar c_{a,b} W_{a,b}^*\right) 
\ee
These are orthogonal because the complex numbers $c_{x,y}$ respect the proof of Proposition \ref{proposition:orthonormal}.
Then, omitting the subscripts, we calculate
\begin{align}
\begin{aligned}
\trace \left(\Psi(F) \right)^2& = |c|^2 \qquad \text{and} \\
\trace \left(\Psi(G) \right)^2 
&= \frac 1 {2n} \trace \left[ c^2 W^2 + 2 |c|^2 I_n + (\bar c)^2 (W^*)^2  \right] 
= |c|^2 = \trace \left(\Psi(H) \right)^2 .
\end{aligned}
\end{align}
This completes the proof. 
\end{proof}
%

Now we recover Theorem 1 of \cite {DFH2006}:
\begin{theorem}\label{theorem:gamma DWCC}
For $\Psi$ in \eqref{eq:DWCC} we have
\be
\gamma(\Phi) =
\begin{cases}  \displaystyle
\frac 1 n + \left( 1 - \frac 1 n \right)  \max_{(a,b) \in \mathcal S_n} \left| c_{a,b}\right|^2 
& \text{if }  \displaystyle\max_{(a,b) \in \mathcal S_n} \left| c_{a,b}\right| \leq 1   \vspace{0.2cm}\\ \displaystyle
 \max_{(a,b) \in \mathcal S_n} \left| c_{a,b}\right|^2 
& \text{if }  \displaystyle \max_{(a,b) \in \mathcal S_n} \left| c_{a,b}\right| > 1
\end{cases} 
\ee
where $c_{a,b}$ is defined in \eqref{eq:phase2} or \eqref{eq:mapped basis}.
Here, $\gamma(\cdot)$ is defined in \eqref{eq:gamma}.
\end{theorem} 
\begin{proof}
By Lemma \ref{lemma:A diagonal}
\be
\|A_\Psi \|_\infty =  \max_{(a,b) \in \mathcal S_n} \left| c_{a,b}\right|^2.
\ee
Applying this to \eqref{eq:gamma} completes the proof. 
\end{proof}

\subsubsection{Complementary channels and their bounds} \label{sec:DWCCC bound}

In this section, we study the following subset of discrete Weyl covariant channels,
whose  complementary channels are unital:
\be\label{eq:DWCC2}
\Psi(\rho) = \frac 1 k \sum_{l=1}^k W_l \rho W_l^* .
\ee
Here, $W_l = W_{x_l,y_l}$ for a $k$-sequence $((x_l,y_l))_{l=1}^k \subseteq \mathbb Z_n \times \mathbb Z_n$
with $1\leq k \leq n^2$; in the sequence there is no multiplicity. 
Their complementary channels are written by
\be\label{eq:DWCCC}
\left[\Psi^C (\rho) \right]_{l,m} =\frac 1k \trace \left[ W_l \rho   W_m^*  \right],
\qquad 1\leq l,m \leq k. 
\ee

Again, our matrix $A_{\Psi^C}$ is diagonal in our basis:
\begin{lemma}\label{lemma:diagonal complementary}
For  $\Psi^C$  in \eqref{eq:DWCCC}, 
the matrix $A_{\Psi^C}$ in \eqref{eq:formula A} is diagonal with respect to the basis in \eqref{eq:basis}
and the diagonal entries are $\{\frac n {2 k^2} \cdot N(a,b) \}_{(a,b)\in \mathcal S_n}$.
Here, 
\begin{align}\begin{aligned}\label{eq:counting}
N(a,b) \equiv & \# \{(l,m)\in \mathbb Z_k \times \mathbb Z_k: \, (x_m-x_l, y_m - y_l)=(a,b) \} \\
& \quad + \# \{(l,m)\in \mathbb Z_k \times \mathbb Z_k: \, (x_m-x_l, y_m - y_l)= (-a,-b)   \}
\end{aligned}\end{align} 
\end{lemma}  
\begin{proof} 
First,  note that $W_{x_m,y_m}^* W_{x_l,y_l} = u_{l,m} W_{x_l-x_m,y_l-y_m}$ 
for some complex number $u_{l,m}$ of modulus $1$,
which we do not specify. 
Then,
\begin{align}\begin{aligned} \label{eq:basis mapped C}
\left[\Psi^C (F_{a,b}) \right]_{l,m} 
&=\frac{\sqrt n}{k}\cdot \left[  v_{l,m,a,b} \cdot \delta_{a,x_m-x_l} \delta_{b, y_m-y_l}\right] \\
\left[\Psi^C (G_{a,b}) \right]_{l,m} 
&= \frac{\sqrt n}{\sqrt 2 \cdot k} \cdot \left[ v_{l,m,a,b} \cdot  \delta_{a,x_m-x_l} \delta_{b, y_m-y_l}+ \bar v_{l,m,a,b} \cdot \delta_{a,x_l-x_m} \delta_{b,y_l-y_m} \right]\\
\left[\Psi^C (H_{a,b}) \right]_{l,m} 
&=\frac{\sqrt n}{\sqrt 2 \cdot k} \cdot \left[ -\mathrm{i} v_{l,m,a,b} \cdot \delta_{a,x_m-x_l} \delta_{b, y_m-y_l}+ \mathrm{i} \bar v_{l,m,a,b} \cdot \delta_{a,x_l-x_m} \delta_{b,y_l-y_m} \right]
\end{aligned}\end{align}
where $v_{l,ma,b}$  is some complex number of modulus $1$, which we do not specify. 
In particular, the above matrices have zero diagonal entries. 
Note that $x_m-x_l=x_l-x_m$ and $ y_m-y_l=y_l-y_m$ imply that $(x_l,y_l) =(x_m,y_m)$ for $(x,y) \in S_n^{GH}$.
This implies that for fixed $l,m$, at least one of the two terms in each of $\left[\Psi^C (G_{a,b}) \right]_{l,m} $ and $\left[\Psi^C (H_{a,b}) \right]_{l,m}$ must vanish; the first term for both or the second for both. 

Let $\Theta = F,G,H$ and $(a,b) \not= (a^\prime,b^\prime)$,
then two matrices $\Psi^C(\Theta_{a,b})$ and $\Psi^C(\Theta_{a^\prime,b^\prime})$
have non-zero elements at different positions with no overlap, which shows 
$\trace [\Psi^C(\Theta_{a,b}) [\Psi^C(\Theta_{a^\prime,b^\prime})]^*]=0$. 
I.e., any two of $n^2 \times n^2$ matrices in \eqref{eq:basis mapped C} are orthogonal to each other
if we choose two different pairs of $(a,b) \in \mathbb Z_n \times \mathbb Z_n$. 
Hence, showing $\trace [\Psi^C(G_{a,b}) [\Psi^C(H_{a,b})]^*]=0$ implies that $A_{\Psi^C}$ is diagonal.
Indeed, 
\begin{align}
\begin{aligned}
\trace [\Psi^C(G_{a,b}) [\Psi^C(H_{a,b})]^*] 
=\frac {n}{2k^2} \cdot \sum_{l \not= k}  
 \left[ \mathrm i \cdot  \delta_{a,x_m-x_l} \delta_{b, y_m-y_l} - \mathrm i \cdot \delta_{a,x_l-x_m} \delta_{b,y_l-y_m} \right]
=\frac {n}{2k^2} \cdot \sum_{l < k}   0 = 0
\end{aligned}
\end{align}

To get diagonal entries of $A_{\Psi^C}$,
\begin{align}
\trace \left[\Psi^C(F_{a,b}) [\Psi^C(F_{a,b})]^*\right]
&= \frac  n {k^2}  \cdot \sum_{1\leq l,m \leq k} 
 \left[  \delta_{a,x_m-x_l} \delta_{b, x_m-y_l} \right] \label{eq:F-delta} \\
\trace \left[\Psi^C(G_{a,b}) [\Psi^C(G_{a,b})]^*\right] 
&=  \frac n {2 k^2} \cdot \sum_{1\leq l,m \leq k} 
 \left[  \delta_{a,x_m-x_l} \delta_{b, u_m-y_l}+ \delta_{a,x_l-x_m} \delta_{b,y_l-y_m} \right] \\
&= \trace \left[\Psi^C(H_{a,b}) [\Psi^C(H_{a,b})]^*\right].
\end{align}
Note that in \eqref{eq:F-delta}, $ \delta_{a,x_m-x_l} \delta_{b, u_m-y_l}= \delta_{a,x_l-x_m} \delta_{b,y_l-y_m}$.
\end{proof} 

\begin{theorem}
For the complementary channel $\Psi^C$ of the channel defined in in \eqref{eq:DWCC}
we have
\be\label{eq:gammaC}
\gamma(\Psi^C) =
\begin{cases}
\displaystyle \frac 1 k + \frac 1 {2k^2} (n-1) \max_{(a,b) \in \mathcal S_n }N(a,b) 
& \text{if } \displaystyle \max_{(a,b) \in \mathcal S_n }N(a,b) \leq 2 \\
 \displaystyle \frac 1 {2k^2} \max_{(a,b) \in \mathcal S_n }N(a,b) 
& \text{if }  \displaystyle\max_{(a,b) \in \mathcal S_n }N(a,b) > 2 
\end{cases} 
\ee
where $N(a,b)$ is defined in \eqref{eq:counting}.
Here, $\gamma(\cdot)$ is defined in \eqref{eq:gamma}.
\end{theorem} 

\subsection{Additive examples by discrete Weyl covariant channels}\label{sec:additive}
In this section, we use Theorem \ref{theorem:additivity} to have some examples of discrete Weyl covariant channels
which show additivity. To do so, we need to look into $\mathcal C_{add}$ in Definition \ref{def:condition}.
Since the preceding paper \cite{DFH2006} has such examples in terms of $\Psi$ in terms of Lemma \ref{lemma:A diagonal},
we construct examples in terms of Lemma \ref{lemma:diagonal complementary}. 
Interestingly, the condition we consider is purely algebraic. 
\begin{example}
For $n=5$, the sequence $(1,2),(2,3),(1,4),(2,4)$ in the definitions \eqref{eq:DWCC2} and \eqref{eq:DWCCC}
gives additive examples $\Psi$ and $\Psi^C$ in the sense of Theorem \ref{theorem:additivity}. 
\end{example} 
\begin{proof}
We prove additivity for $\Phi^C$, then additivity is also true for $\Phi$ by \eqref{eq:complementN} and \eqref{eq:complementE}. 
First, we draw the table of $(x,y)-(x^\prime,y^\prime)$ where $(x,y)$ name columns and $(x^\prime,y^\prime)$ rows. 
\begin{center}
\begin{tabular}{|l||l|l|l|l|} 
\hline
&(1,2)&(2,3)&(1,4)&(2,4) \\ \hline \hline
(1,2)&(0,0)&(1,1)&(0,2)&(1,2)\\ \hline
(2,3)&(4,4)&(0,0)&(4,1)&(0,1) \\ \hline
(1,4)&(0,3)&(1,4)&(0,0)&(1,0)\\ \hline
(2,4)&(4,3)&(0,4)&(4,0)&(0,0)\\ \hline
\end{tabular}
\end{center}
Then, $N(a,b)$ in \eqref{eq:counting} becomes as follows:
\be
N(a,b) = \begin{cases}
2 &\text{ if } (a,b)=(1,2),(1,4),(1,1),(0,1),(0,2),(1,0) \\
0 &\text{ if } (a,b)=(1,3),(2,3),(2,4),(3,4),(2,2),(2,0).
\end{cases} 
\ee
Note that $(a,b)$'s which give $2$ constitute the eigenspaces for the largest eigenvalues of $A_{\Psi}$.
Among them, we pick up $(0,1)$ and $(0,2)$ so that 
$G_{0,1},H_{0,1},G_{0,2},H_{0,2}$ generate $V, V^2,V^3,V^4$ where $V$ is defined in \eqref{eq:UV}. 
Since $\frac 15[I+V+V^2+V^3+V^4]$ is a rank-one projection, 
the channel $\Psi^C$ satisfies $\mathcal C_{add}$, so that Theorem \ref{theorem:additivity} can be applied. 
\end{proof}

\section{Norms of super operators}\label{sec:pq norm} 
In this section, we reinterpret our bound on $\| \cdot \|_{p\to2}$ of Lemma \ref{lemma:bound1} in the complex settings.
Up to now we focused on linear maps $\Phi:H_n \to H_k$, but our results are compatible 
with maps $\Phi:  M_n(\mathbb C) \to M_k(\mathbb C)$. 
In fact, an orthonormal basis $(M_i)_{i=0}^{n^2-1}$ in $H_n$ becomes one in $M_n(\mathbb C)$ 
if we span them over $\mathbb C$. 
To pursue this problem, we introduce another norm for $\Phi$:
\be\label{eq:pq norm c}
|\Phi |_{p\to q} = \max_{0 \not = \rho\in M_n(\mathbb C)} \frac {\|\Phi(\rho)\|_q}{\|\rho\|_p}
\ee
In general, $\|\Phi \|_{p\to q} \leq |\Phi |_{p\to q}$, but the equality holds when $\Phi$ is completely positive \cite{Watrous2005}.

\begin{theorem}
Take linear maps $\Phi$ and $\Omega$ where $\Phi$ is unital and trace-preserving. 
Then, for $1\leq p \leq 2$
\eq{
|\Phi\otimes\Omega|_{p\to2}^2 \leq
\begin{cases}
\displaystyle k^{-1}n^{2-\frac 2p} + \left(1-n^{1-\frac 2p}\right) \|A_\Phi\|_\infty \cdot |\Omega|_{p \to 2}^2
& \text{if } k \cdot \|A_\Phi\|_\infty \leq  n  \vspace{1mm}\\
\displaystyle \|A_\Phi\|_\infty \cdot |\Omega|_{p \to 2}^2 
&\text{if } k \cdot  \|A_\Phi\|_\infty > n.
\end{cases} 
}
\end{theorem} 
\begin{proof}
First, Lemma \ref{lemma:bound1} and \eqref{eq:bound2} can naturally be extended to our case. 
Since $\|\rho \|_p=1$ we have
\be
\trace \left[\left(\mathrm{id} \otimes \Omega (\rho) \right)^2\right] \leq |\mathrm{id} \otimes \Omega|_{p\to2}^2
=  | \Omega|_{p\to2}^2
\ee
The equality was proved in \cite{Watrous2005}.

Next, if $\|A_\Phi\|_\infty \geq \frac nk $ then the first term in \eqref{eq:bound2} is non-positive,
which we ignore to get the bound $\|A_\Phi\|_\infty \cdot |\Omega|_{p \to 2}^2$. 
Otherwise we bound the first term by
\eq{
\trace \left[\Omega(\rho_k)^2\right] 
\leq n^{2-\frac 2p} |\Omega|_{p\to2}^2.
}
because $\|\rho_K\|_p \leq n^{1-\frac 1p}$. 
This completes the proof. 
\end{proof}

\emph{Acknowledgments:---}
G.G. research is supported by NSERC. 
 M.F. was financially supported by the CHIST-ERA/BMBF project CQC and the John
Templeton Foundation (ID\#48322).

\bibliographystyle{alpha}
\bibliography{ClosedFormula-bib}

\end{document}